\documentclass[12pt,envcountsame,envcountsect,runningheads,fleqn]{llncs}

\usepackage[a4paper,margin=2.52cm]{geometry}

\newcommand{\pbLength}{10cm}

\newcommand{\LongVersion}[1]{#1}
\newcommand{\ShortVersion}[1]{}

\def\tuple#1{\langle#1\rangle}
\usepackage{latexsym}
\usepackage{xspace}
\usepackage{amssymb}
\usepackage{stmaryrd}
\usepackage{url}
\usepackage{oldgerm}
\usepackage{graphicx}
\input{xy}
\xyoption{all}
\usepackage[ruled,vlined,linesnumbered]{algorithm2e}

\def\eqref#1{(\ref{#1})}
\def\eqrefp#1{(\ref{#1}')}

\newcommand{\mand}{\sqcap}
\newcommand{\mor}{\sqcup}
\newcommand{\V}{\forall}
\newcommand{\E}{\exists}

\newcommand{\mL}{\mathcal{L}}

\newcommand{\mI}{\mathcal{I}}

\newcommand{\mZ}{\mathcal{Z}}

\newcommand{\CN}{\Sigma_C}
\newcommand{\RN}{\Sigma_R}
\newcommand{\RNpm}{\Sigma_R^\pm}
\newcommand{\IN}{\Sigma_I}

\newcommand{\Self}{\mathtt{Self}}

\newcommand{\QU}{\mathtt{Q}}
\newcommand{\SE}{\mathtt{S}}

\newcommand{\mLP}{\mL_\Phi}
\newcommand{\mLPp}{\mL_\Phi^\mathit{pos}}
\newcommand{\mLPpE}{\mL_{\Phi,\E}^\mathit{pos}}
\newcommand{\mLPpV}{\mL_{\Phi,\V}^\mathit{pos}}
\newcommand{\mLPsp}{\mL_\Phi^\mathit{sp}}
\newcommand{\mLPspE}{\mL_{\Phi,\E}^\mathit{sp}}
\newcommand{\mLPspV}{\mL_{\Phi,\V}^\mathit{sp}}
\newcommand{\simLP}{{\sim_{\,\Phi,\mI}}}
\newcommand{\equivLP}{{\equiv_{\,\Phi,\mI}}}

\newcommand{\leqPp}{\leq_\Phi^\mathit{pos}}
\newcommand{\leqPsp}{\leq_\Phi^\mathit{sp}}
\newcommand{\equivPsp}{\equiv^\mathit{sp}_{\Phi}}
\newcommand{\equivP}{\equiv_{\Phi}}

\newcommand{\ALCreg}{$\mathcal{ALC}_{reg}$\xspace}

\def\ramka#1{\begin{center}
\fbox{\parbox{11.90cm}{\begin{center}#1\end{center}}}
\end{center}}

\def\trojkat{\mbox{{\scriptsize$\!\vartriangleleft$}}}

\newcommand{\myend}{\mbox{}\hfill\trojkat}

\newcommand{\comment}[1]{}

\newcommand{\TidyPsp}{\mathsf{Tidy}^{sp}_\Phi}

\title{Bisimulation-Based Comparisons\\ for Interpretations in Description Logics}
\titlerunning{Bisimulation-Based Comparisons for Interpretations}

\author{Ali Rezaei Divroodi \and Linh Anh Nguyen}
\institute{
Institute of Informatics, University of Warsaw\\
Banacha 2, 02-097 Warsaw, Poland\\
\email{\{rezaei,nguyen\}@mimuw.edu.pl}
}

\authorrunning{A.R. Divroodi and L.A. Nguyen}

\tocauthor{Ali Rezaei Divroodi (University of Warsaw),
	   Linh Anh Nguyen (University of Warsaw)}

\begin{document}
\maketitle
\sloppy

\begin{abstract}
We study comparisons between interpretations in description logics with respect to ``logical consequences'' of the form of semi-positive concepts (like semi-positive concept assertions). Such comparisons are characterized by conditions similar to the ones of bisimulations. The simplest among the considered logics is a variant of PDL (propositional dynamic logic). The others extend that logic with inverse roles, nominals, quantified number restrictions, the universal role, and/or the concept constructor for expressing the local reflexivity of a role. The studied problems are: preservation of semi-positive concepts with respect to comparisons, the Hennessy-Milner property for comparisons, and minimization of interpretations that preserves semi-positive concepts.
\end{abstract}

\section{Introduction}

Bisimulation is a natural notion of equivalence arose in modal logic~\cite{vBenthem76,vBenthem83,vBenthem84} and state transition systems~\cite{Park81,HennessyM85}. It can be viewed as a binary relation associating state transition systems which behave in the same way in the sense that one system simulates the other and vice versa. Kripke models in modal logic are a special case of labeled state transition systems. 

Bisimulations have widely been studied for various variants of modal logic like dynamic logic, temporal logic, hybrid logic and, in particular, also for description logics (DLs) \cite{KurtoninaR99,BSDL-CSP,LutzPW11}. 
They have been used for analyzing the expressivity of a wide range of modal logics (see, e.g., \cite{BRV2001} for details), for minimizing state transition systems, as well as for concept learning in DLs (e.g., \cite{LbRoughification,KSE2012,SoICT2012,DivroodiHNN12}). 

Bisimilarity between two states is usually defined by three conditions (the states have the same label, each transition from one of the states can be simulated by a similar transition from the other, and vice versa). For bisimulation between two pointed-models, the initial states of the models are also required to be bisimilar. When converse is allowed, two additional conditions are required for bisimulation~\cite{BRV2001}. Bisimulation conditions for dealing with graded modalities were studied in~\cite{Rijke00,ConradieThesis,JaninL04}. In the field of hybrid logic, the bisimulation condition for dealing with nominals is well known (see, e.g., \cite{ArecesBM01}). In DLs, such conditions are used for dealing with inverse roles, (quantified) number restrictions and nominals, respectively. There are also bisimulation conditions for dealing with individuals, the universal role and the $\Self$ constructor in DLs~\cite{BSDL-arxiv,KSE2012}. 

In modal logic, bisimulation invariance has the form: if two states are bisimilar then they satisfy the same set of formulas (i.e., all modal formulas are invariant w.r.t.\ bisimulation). For the converse, the Hennessy-Milner property states that, in finitely branching Kripke models, two states are bisimilar iff they satisfy the same set of formulas. This property can be generalized for non-finitely branching Kripke models (see, e.g., \cite{LutzPW11}). 

Simulation is a notion with weaker conditions than bisimulation. It is only ``one way'', while bisimulation is ``two way''. In the most common understanding, the ``ways'' are related with the ``transitions'' but not w.r.t. comparison between the sets of atomic formulas satisfied at the considered states. Such simulation preserves positive existential formulas (see, e.g., \cite{BRV2001}).

What variant of bisimulation can be used to talk about preservation of positive formulas, which may use both existential and universal modal operators? Defining positive formulas to be the ones without $\bot$ (falsity), $\lnot$ (negation) and $\to$ (implication), in~\cite{nguyen00least} Nguyen gave a bisimulation-based comparison between Kripke models that preserves positive formulas in basic serial monomodal logics. In~\cite{nguyen08igpl} he extended the preservation result also for serial regular grammar logics and proved the corresponding Hennessy-Milner property. Such bisimulation-based comparison uses the conditions of bisimulation for ``transitions'' and compares the sets of atomic formulas satisfied at the considered states. Bisimulation-based comparison between Kripke models is worth studying, because it can be used for minimizing a Kripke model w.r.t.\ the set of logical consequences being positive formulas. For example, after constructing a least Kripke model of a positive modal logic program in a serial modal logic~\cite{nguyen00least,nguyen08igpl,KNS2010jar}, one can minimize it w.r.t.\ positive formulas to obtain a minimal Kripke model that characterizes the program w.r.t. positive consequences. Such minimization is also applicable to (non-serial) DLs~\cite{nguyen06jelia,RG-Nguyen}. 

In this paper, we study bisimulation-based comparisons between interpretations in DLs. The simplest among the considered logics is \ALCreg, a variant of PDL (propositional dynamic logic). The others extend that logic with inverse roles, nominals, quantified number restrictions, the universal role, and/or the concept constructor for expressing the local reflexivity of a role. The studied problems are: preservation of semi-positive concepts with respect to comparisons, the Hennessy-Milner property for comparisons, and minimization of interpretations that preserves semi-positive concepts. The class of semi-positive concepts differs from the class of positive concepts in that, in the recursive definition, it allows also $\bot$. This is involved with non-seriality. 

Apart from~\cite{nguyen00least,nguyen08igpl,KNS2010jar}, bisimulation-based comparisons for modal logics were studied also in~\cite{abs-1303-2467} (and possibly other works). In~\cite{abs-1303-2467} the notion is studied at an abstract level for coalgebraic modal logics under the name $\Lambda$-simulation, and the term ``positive formula'' is used instead of ``semi-positive formula''. As mentioned before, the term ``simulation'' traditionally has another meaning, and in our opinion $\bot$ should not be referred to as ``positive''. At an abstract level, the work~\cite{abs-1303-2467} does not have a result like a Hennessy-Milner property. In the current work, to guarantee a Hennessy-Milner property, roles in semi-positive concepts have a specific syntax due to the presence of the test operator. The definition of semi-positive concepts itself in the current work is not trivial (e.g., we have that if $C$ is a semi-positive concept then $\leq\!n\,r.\lnot C$ is also a positive concept). 

Our results on preservation of semi-positive concepts and the Hennessy-Milner property w.r.t. comparisons may overlap to a certain degree with the known ones (we will carefully check this later). However, our results on ``characterizing bisimulation by semi-positive concepts'' and ``minimization preserving semi-positive concepts'' are completely novel. 


\section{Notation and Semantics of Description Logics}

Our languages use a finite set $\CN$ of {\em concept names} (atomic concepts), a finite set $\RN$ of {\em role names} (atomic roles), and a finite set $\IN$ of {\em individual names}. Let $\Sigma = \CN \cup \RN \cup \IN$. We denote concept names by letters like $A$ and $B$, denote role names by letters like $r$ and $s$, and denote individual names by letters like $a$ and $b$.

We consider some (additional) {\em DL-features} denoted by $I$ ({\em inverse}), $O$ ({\em nominal}), $Q$ ({\em quantified number restriction}), $U$ ({\em universal role}), $\Self$. A {\em set of DL-features} is a set consisting of some or zero of these names.

Let $\Phi$ be any set of DL-features and let $\mL$ stand for \ALCreg. The DL language $\mLP$ allows {\em roles} and {\em concepts} defined inductively as follows:
\begin{itemize}
\item if $r \in \RN$ then $r$ is a role of $\mLP$
\item if $A \in \CN$ then $A$ is a concept of $\mLP$
\item if $R$ and $S$ are roles of $\mLP$ and $C$ is a concept of $\mLP$ then
   \begin{itemize}
   \item $\varepsilon$, $R \circ S$ , $R \sqcup S$, $R^*$ and $C?$ are roles of $\mLP$
   \item $\top$, $\bot$, $\lnot C$, $C \mor D$, $C \mand D$, $\E R.C$ and $\V R.C$ are concepts of $\mLP$
   \item if $I \in \Phi$ then $R^-$ is a role of $\mLP$
   \item if $O \in \Phi$ and $a \in \Sigma_I$ then $\{a\}$ is a concept of $\mLP$
   \item if $Q \in \Phi$, $r \in \RN$ and $n$ is a natural number\\ then $\geq n\,r.C$ and $\leq n\,r.C$ are concepts of $\mLP$
   \item if $\{Q,I\} \subseteq \Phi$, $r \in \RN$ and $n$ is a natural number\\ then $\geq n\,r^-.C$ and $\leq n\,r^-.C$ are concepts of $\mLP$
   \item if $U \in \Phi$ then $U$ is a role of $\mLP$ 
   \item if $\Self \in \Phi$ and $r \in \RN$ then $\E r.\Self$ is a concept of $\mLP$.
   \end{itemize}
\end{itemize}

We use letters like $R$ and $S$ to denote arbitrary roles, and use letters like $C$ and $D$ to denote arbitrary concepts.
\LongVersion{A role stands for a binary relation, while a concept stands for a unary relation.

The intended meaning of the role constructors is the following:
\begin{itemize}
\item $R \circ S$ stands for the sequential composition of $R$ and $S$,
\item $R \sqcup S$ stands for the set-theoretical union of $R$ and $S$,
\item $R^*$ stands for the reflexive and transitive closure of $R$,
\item $C?$ stands for the test operator (as of PDL),
\item $R^-$ stands for the {\em inverse} of $R$.
\end{itemize}

The concept constructors $\E R.C$ and $\V R.C$ correspond respectively to the modal operators $\tuple{R}C$ and $[R]C$ of PDL.
The concept constructors $\geq n\,R.C$ and $\leq n\,R.C$ are called {\em quantified number restrictions}. They correspond to graded modal operators.

} 
We refer to elements of $\RN$ also as {\em atomic roles}. Let $\RNpm = \RN \cup \{r^- \mid r \in \RN\}$. From now on, by {\em basic roles} we refer to elements of $\RNpm$ if the considered language allows inverse roles, and refer to elements of $\RN$ otherwise. In general, the language decides whether inverse roles are allowed in the considered context.

An {\em interpretation} $\mI = \langle \Delta^\mI, \cdot^\mI \rangle$ consists of a non-empty set $\Delta^\mI$, called the {\em domain} of $\mI$, and a function $\cdot^\mI$, called the {\em interpretation function} of $\mI$, which maps every concept name $A$ to a subset $A^\mI$ of $\Delta^\mI$, maps every role name $r$ to a binary relation $r^\mI$ on $\Delta^\mI$, and maps every individual name $a$ to an element $a^\mI$ of~$\Delta^\mI$.
The interpretation function $\cdot^\mI$ is extended to complex roles and complex concepts as shown in Figure~\ref{fig: int-comp}, where $\#\Gamma$ stands for the cardinality of the set $\Gamma$.
We write $C^\mI(x)$ to denote $x \in C^\mI$, and write $R^\mI(x,y)$ to denote $\tuple{x,y} \in R^\mI$.

An interpretation $\mI$ is said to be {\em serial} in $\mLP$ if, for every basic role $R$ of $\mLP$ and every $x \in \Delta^\mI$, there exists $y \in \Delta^\mI$ such that $\tuple{x,y} \in R^\mI$.   

\begin{figure}[t]
\ramka{ \(
\begin{array}{c}
\begin{array}{rcl}
(R \circ S)^\mI & = & R^\mI \circ S^\mI \\[0.5ex]
(R \sqcup S)^\mI & = & R^\mI \cup S^\mI \\[0.5ex]
(R^*)^\mI & = & (R^\mI)^* \\[0.5ex]
(C?)^\mI & = & \{ \tuple{x,x} \mid C^\mI(x) \} \\[0.5ex]
\varepsilon^\mI & = & \{\tuple{x,x} \mid x \in \Delta^\mI\} \\[0.5ex]
U^\mI & = & \Delta^\mI \times \Delta^\mI \\[0.5ex]
(R^-)^\mI & = & (R^\mI)^{-1}
\end{array}
\quad\quad\quad
\begin{array}{rcl}
\top^\mI & = & \Delta^\mI \\[0.5ex]
\bot^\mI & = & \emptyset \\[0.5ex]
(\lnot C)^\mI & = & \Delta^\mI \setminus C^\mI \\[0.5ex]
(C \mor D)^\mI & = & C^\mI \cup D^\mI \\[0.5ex]
(C \mand D)^\mI & = & C^\mI \cap D^\mI \\[0.5ex]
\{a\}^\mI & = & \{a^\mI\} \\[0.5ex]
(\E r.\Self)^\mI & = & \{x \in \Delta^\mI \mid r^\mI(x,x)\}
\end{array} \\
\\[-1ex]
\begin{array}{rcl}
(\E R.C)^\mI & = & \{ x \in \Delta^\mI \mid \E y\,[R^\mI(x,y) \textrm{ and } C^\mI(y)] \\[1ex]
(\V R.C)^\mI & = & \{ x \in \Delta^\mI \mid \V y\,[R^\mI(x,y) \textrm{ implies } C^\mI(y)] \} \\[1ex]
(\geq n\,R.C)^\mI & = & \{x \in \Delta^\mI \mid \#\{y \mid R^\mI(x,y) \textrm{ and } C^\mI(y)\} \geq n \} \\[1ex]
(\leq n\,R.C)^\mI & = & \{x \in \Delta^\mI \mid \#\{y \mid R^\mI(x,y) \textrm{ and } C^\mI(y)\} \leq n \}
\end{array}
\end{array}
\) }\vspace{-0.5em}
\caption{Interpretation of complex roles and complex concepts.}
\label{fig: int-comp}
\end{figure}

We say that a role $R$ is in the {\em converse normal form} (CNF) if the inverse constructor is applied in $R$ only to role names and the role $U$ is not under the scope of any other role constructor. Since every role can be translated to an equivalent role in CNF,\footnote{For example, $((r \mor s^-)\circ r^*)^- = (r^-)^* \circ (r^- \mor s)$.} in this paper we assume that roles are presented in the CNF.


\section{Positive and Semi-Positive Concepts}
\label{section: poscon}

Let $\mLPp$ be the smallest set of concepts and $\mLPpE$, $\mLPpV$ be the smallest sets of roles defined recursively as follows:
\begin{itemize}
\item if $r \in \RN$ then $r$ is a role of $\mLPpE$ and $\mLPpV$,
\item if $I \in \Phi$ and $r \in \RN$ then $r^-$ is a role of $\mLPpE$ and $\mLPpV$,
\item if $R$ and $S$ are roles of $\mLPpE$ and $C$ is a concept of $\mLPp$\\ then $\varepsilon$, $R \circ S$ , $R \sqcup S$, $R^*$ and $C?$ are roles of $\mLPpE$,
\item if $R$ and $S$ are roles of $\mLPpV$ and $C$ is a concept of $\mLPp$\\ then $\varepsilon$, $R \circ S$ , $R \sqcup S$, $R^*$ and $(\lnot C)?$ are roles of $\mLPpV$,
\item if $A \in \CN$ then $A$ is a concept of $\mLPp$,
\item if $O \in \Phi$ and $a \in \Sigma_I$ then $\{a\}$ is a concept of $\mLPp$,
\item if $\Self \in \Phi$ and $r \in \RN$ then $\E r.\Self$ is a concept of $\mLPp$,
\item if $C$ is a concept of $\mLPp$, $R$ is a role of $\mLPpE$ and $S$ is a role of $\mLPpV$ then
   \begin{itemize}
   \item $\top$, $C \mor D$, $C \mand D$, $\E R.C$ and $\V S.C$ are concepts of $\mLPp$,
   \item if $Q \in \Phi$, $r \in \RN$ and $n$ is a natural number\\ then $\geq n\,r.C$ and $\leq n\,r.(\lnot C)$ are concepts of $\mLPp$,
   \item if $\{Q,I\} \subseteq \Phi$, $r \in \RN$ and $n$ is a natural number\\ then $\geq n\,r^-.C$ and $\leq n\,r^-.(\lnot C)$ are concepts of $\mLPp$,
   \item if $U \in \Phi$ then $\V U.C$ and $\E U.C$ are concepts of $\mLPp$.
   \end{itemize}
\end{itemize}

A concept of $\mLPp$ is called a {\em positive concept} of $\mLP$.
We introduce both $\mLPpV$ and $\mLPpE$ due to the test constructor of roles. The concepts $\E (A?).B$ and $\V ((\lnot A)?).B$ are positive concepts; they are equivalent to $A \mand B$ and $A \mor B$, respectively. That the concept $\leq\!n\,R.(\lnot A)$ is positive should not be a surprise, as $\V R.A$ is equivalent to $\leq\!0\,R.(\lnot A)$.

Let $\mLPsp$ be the smallest set of concepts and $\mLPspE$, $\mLPspV$ be the smallest sets of roles defined analogously to the case of $\mLPp$, $\mLPpE$, $\mLPpV$ except that $\bot$ is also allowed as a concept of $\mLPsp$. 
We call concepts of $\mLPsp$ {\em semi-positive concepts} of~$\mLP$.

\section{Bisimulation-Based Comparisons for Interpretations}
\label{section: bis-inv}

Let $\mI$ and $\mI'$ be interpretations.
A binary relation $Z \subseteq \Delta^\mI \times \Delta^{\mI'}$ is called an {\em $\mLP$-comparison} between $\mI$ and $\mI'$ if the following conditions hold for every $a \in \Sigma_I$, $A \in \CN$, $r \in \RN$, $x,y \in \Delta^\mI$, $x',y' \in \Delta^{\mI'}\,$:
\begin{eqnarray}
&& Z(a^\mI,a^{\mI'}) \label{bs:eqA} \\
&& Z(x,x') \Rightarrow [A^\mI(x) \Rightarrow A^{\mI'}(x')] \label{bs:eqB} \\
&& [Z(x,x') \land r^\mI(x,y)] \Rightarrow \E y' \in \Delta^{\mI'}[Z(y,y') \land r^{\mI'}(x',y')] \label{bs:eqC} \\
&& [Z(x,x') \land r^{\mI'}(x',y')] \Rightarrow \E y \in \Delta^\mI [Z(y,y') \land r^\mI(x,y)], \label{bs:eqD}
\end{eqnarray}
if $I \in \Phi$ then
\begin{eqnarray}
&& [Z(x,x') \land r^\mI(y,x)] \Rightarrow \E y' \in \Delta^{\mI'}[Z(y,y') \land r^{\mI'}(y',x')] \label{bs:eqI1} \\
&& [Z(x,x') \land r^{\mI'}(y',x')] \Rightarrow \E y \in \Delta^\mI [Z(y,y') \land r^\mI(y,x)], \label{bs:eqI2}
\end{eqnarray}
if $O \in \Phi$ then
\begin{eqnarray}
&& Z(x,x') \Rightarrow [x = a^\mI \Rightarrow x' = a^{\mI'}], \label{bs:eqO0}
\end{eqnarray}
if $Q \in \Phi$ then
\begin{eqnarray}
&& \parbox{\pbLength}{if $Z(x,x')$ holds then, for every role name $r$, there exists a bijection $h: \{y \mid r^\mI(x,y)\} \to \{y' \mid r^{\mI'}(x',y')\}$ such that $h \subseteq Z$,} \label{bs:eqQ}
\end{eqnarray}
if $\{Q,I\} \subseteq \Phi$ then (additionally)
\begin{eqnarray}
&& \parbox{\pbLength}{if $Z(x,x')$ holds then, for every role name $r$, there exists a bijection $h: \{y \mid r^\mI(y,x)\} \to \{y' \mid r^{\mI'}(y',x')\}$ such that $h \subseteq Z$,} \label{bs:eqQI}
\end{eqnarray}
if $U \in \Phi$ then
\begin{eqnarray}
&& \V x \in \Delta^\mI\, \E x' \in \Delta^{\mI'}\, Z(x,x') \label{bs:eqU1} \\
&& \V x' \in \Delta^{\mI'}\, \E x \in \Delta^\mI\, Z(x,x'), \label{bs:eqU2}
\end{eqnarray}
if $\Self \in \Phi$ then
\begin{eqnarray}
&& Z(x,x') \Rightarrow [r^\mI(x,x) \Rightarrow r^{\mI'}(x',x')]. \label{bs:eqSelf}
\end{eqnarray}

For example, if $\Phi = \{Q,I\}$ then only the conditions \eqref{bs:eqA}-\eqref{bs:eqI2}, \eqref{bs:eqQ} and \eqref{bs:eqQI} (and all of them) are essential.

By \eqrefp{bs:eqB}, \eqrefp{bs:eqO0}, \eqrefp{bs:eqSelf} we denote the conditions obtained respectively from \eqref{bs:eqB}, \eqref{bs:eqO0}, \eqref{bs:eqSelf} by replacing the second implication ($\Rightarrow$) by equivalence ($\Leftrightarrow$). 
If the conditions \eqref{bs:eqB}, \eqref{bs:eqO0}, \eqref{bs:eqSelf} are replaced by \eqrefp{bs:eqB}, \eqrefp{bs:eqO0}, \eqrefp{bs:eqSelf} then the relation $Z$ is called an {\em $\mLP$-bisimulation} between $\mI$ and~$\mI'$~\cite{BSDL-arxiv}. 


\begin{proposition} \label{prop:pr-bis}
\
\begin{enumerate}
\item \label{pr-bis-1} The relation $\{\tuple{x,x} \mid x \in \Delta^\mI\}$ is an $\mLP$-comparison between $\mI$ and $\mI$.
\item \label{pr-bis-3} If $Z_1$ is an $\mLP$-comparison between $\mI_0$ and $\mI_1$, and $Z_2$ is an $\mLP$-comparison between $\mI_1$ and $\mI_2$, then $Z_1 \circ Z_2$ is an $\mLP$-comparison between $\mI_0$ and $\mI_2$.
\item \label{pr-bis-4} If $\mZ$ is a set of $\mLP$-comparison between $\mI$ and $\mI'$ then $\bigcup\mZ$ is also an $\mLP$-comparison between $\mI$ and $\mI'$.
\end{enumerate}
\end{proposition}

The proof of this proposition is straightforward.

\begin{lemma} \label{lemma: bs-inv-1}
Let $\mI$ and $\mI'$ be interpretations and $Z$ be an $\mLP$-comparison between $\mI$ and $ \mI'$. Then the following properties hold for every concept $C$ of $\mLPsp$, every role $R$ of $\mLPspE$, every role $S$ of $\mLPspV$, every $x, y \in \Delta^\mI$, every $x', y' \in \Delta^{\mI'}$, and every $a \in \mI$:
\begin{eqnarray}
&& Z(x,x') \Rightarrow [C^\mI(x) \Rightarrow C^{\mI'}(x')] \label{bs:eqB-2} \\
&& [Z(x,x') \land R^\mI(x,y)] \Rightarrow \E y' \in \Delta^{\mI'}[Z(y,y') \land R^{\mI'}(x',y')] \label{bs:eqC-2} \\
&& [Z(x,x') \land S^{\mI'}(x',y')] \Rightarrow \E y \in \Delta^\mI [Z(y,y') \land S^\mI(x,y)] \label{bs:eqD-2}.
\end{eqnarray}
\end{lemma}

See the appendix for a proof of this lemma.


A~concept $C$ of $\mLP$ is said to be {\em preserved by $\mLP$-comparisons} if, for any interpretations $\mI$, $\mI'$ and any $\mLP$-comparison $Z$ between $\mI$ and $\mI'$, if $Z(x,x')$ holds and $x \in C^\mI$ then $x' \in C^{\mI'}$.
The following theorem follows immediately from the assertion~\eqref{bs:eqB-2} of Lemma~\ref{lemma: bs-inv-1}.

\begin{theorem} \label{theorem: bs-inv-1}
All concepts of $\mLPsp$ are preserved by $\mLP$-comparisons.
\end{theorem}

\begin{corollary} \label{cor: bs-inv-2}
All concepts of $\mLPp$ are preserved by $\mLP$-comparisons.
\end{corollary}


Let $\mI$ and $\mI'$ be interpretations, $x \in \Delta^\mI$ and $x' \in \Delta^{\mI'}$. Define that: 
\begin{itemize}
\item $x$ is {\em equivalent to} $x'$ w.r.t.\ (concepts of) $\mLP$, denoted by $x \equivP x'$, if, for every concept $C$ of $\mLP$, $x \in C^\mI$ iff $x' \in C^{\mI'}$;
\item $x$ is {\em less than or equal to} $x'$ w.r.t.\ concepts of $\mLPsp$ (resp.~$\mLPp$), denoted by $x \leqPsp x'$ (resp.~$x \leqPp x'$), if, for every concept $C$ of $\mLPsp$ (resp.~$\mLPp$), $x \in C^\mI$ implies $x' \in C^{\mI'}$; 
\item $x$ is {\em equivalent to} $x'$ w.r.t.\ concepts of $\mLPsp$, denoted by $x \equivPsp x'$, if $x \leqPsp x'$ and $x' \leqPsp x$. 
\end{itemize} 

We say that an interpretation $\mI$ is {\em finitely branching} (or {\em image-finite}) w.r.t.\ $\mLP$ if, for every $x \in \Delta^\mI$ and every basic role $R$ of $\mLP$, the set $\{y \in \Delta^\mI \mid R^\mI(x,y)\}$ is finite. 
We say that $\mI$ is {\em unreachable-objects-free} (w.r.t.\ $\mLP$) if every element of $\Delta^\mI$ is reachable from some $a^\mI$ (with $a \in \IN$) via a path consisting of edges being instances of basic roles (of~$\mLP$).
The following theorem comes from our work~\cite{BSDL-arxiv}.

\begin{theorem}[The Hennessy-Milner Property] \label{theorem: H-M-0}
Let $\mI$ and $\mI'$ be finitely branching interpretations (w.r.t.~$\mLP$) such that, for every $a \in \IN$, $a^\mI \equivP a^{\mI'}$. 
Suppose that if $U \in \Phi$ then either $\IN \neq \emptyset$ and both $\mI$, $\mI'$ are finite, or both $\mI$, $\mI'$ are unreachable-objects-free. 
Then, for every $x \in \Delta^\mI$ and $x' \in \Delta^{\mI'}$, $x \equivP x'$ iff there exists an $\mLP$-bisimulation $Z$ between $\mI$ and $\mI'$ such that $Z(x,x')$ holds. In particular, the relation $\{\tuple{x,x'} \in \Delta^ \mI\times \Delta^{\mI'} \mid x \equivP x'\}$ is an $\mLP$-bisimulation between $\mI$ and~$\mI'$.
\end{theorem}


In the rest of this section we present theorems similar to the Hennessy-Milner property that are related to $\mLP$-comparisons and/or semi-positive concepts.

\begin{theorem} \label{theorem: H-M}
Let $\mI$ and $\mI'$ be finitely branching interpretations (w.r.t.~$\mLP$) such that, for every $a \in \IN$, $a^\mI \leqPsp a^{\mI'}$. 
Suppose that if $U \in \Phi$ then either $\IN \neq \emptyset$ and both $\mI$, $\mI'$ are finite, or both $\mI$, $\mI'$ are unreachable-objects-free. 
Then, for every $x \in \Delta^\mI$ and $x' \in \Delta^{\mI'}$, $x \leqPsp x'$ iff there exists an $\mLP$-comparison $Z$ between $\mI$ and $\mI'$ such that $Z(x,x')$ holds. In particular, the relation $\{\tuple{x,x'} \in \Delta^ \mI\times \Delta^{\mI'} \mid x \leqPsp x'\}$ is an $\mLP$-comparison between $\mI$ and~$\mI'$.
\end{theorem}

See the appendix for a proof of this theorem.


Analyzing the proof of Theorem~\ref{theorem: H-M}, it can be seen that, in the case $Q \notin \Phi$, $\bot$ is only used for showing that there exists $y \in \Delta^\mI$ such that $r^\mI(x,y)$ holds when proving the condition~\eqref{bs:eqD}. If $\mI$ is a serial interpretation then that property is guaranteed. Therefore, we also have the following theorem, whose proof is very similar to the proof of Theorem~\ref{theorem: H-M}.  

\begin{theorem} \label{theorem: H-M-2}
Let $\mI$ and $\mI'$ be finitely branching interpretations (w.r.t.~$\mLP$) such that $\mI$ is serial and, for every $a \in \IN$, $a^\mI \leqPp a^{\mI'}$. 
Suppose $Q \notin \Phi$ and if $U \in \Phi$ then either $\IN \neq \emptyset$ and both $\mI$, $\mI'$ are finite, or both $\mI$, $\mI'$ are unreachable-objects-free. 
Then, for every $x \in \Delta^\mI$ and $x' \in \Delta^{\mI'}$, $x \leqPp x'$ iff there exists an $\mLP$-comparison $Z$ between $\mI$ and $\mI'$ such that $Z(x,x')$ holds. In particular, the relation $\{\tuple{x,x'} \in \Delta^ \mI\times \Delta^{\mI'} \mid x \leqPp x'\}$ is an $\mLP$-comparison between $\mI$ and~$\mI'$.
\end{theorem}


\section{Characterizing Bisimulation by Semi-Positive Concepts}

In the case $Q \in \Phi$, there is a closer relationship between semi-positive concepts and $\mLP$-bisimulation from the semantic point of view. 

\begin{theorem} \label{theorem: H-M-3}
Let $\mI$ and $\mI'$ be finitely branching interpretations (w.r.t.~$\mLP$) such that, for every $a \in \IN$, $a^\mI \equivPsp a^{\mI'}$. 
Suppose $Q \in \Phi$ and if $U \in \Phi$ then both $\mI$ and $\mI'$ are unreachable-objects-free. 
Then, for every $x \in \Delta^\mI$ and $x' \in \Delta^{\mI'}$, $x \equivPsp x'$ iff there exists an $\mLP$-bisimulation $Z$ between $\mI$ and $\mI'$ such that $Z(x,x')$ holds. In particular, the relation $\{\tuple{x,x'} \in \Delta^ \mI\times \Delta^{\mI'} \mid x \equivPsp x'\}$ is an $\mLP$-bisimulation between $\mI$ and~$\mI'$.
\end{theorem}

See the appendix for a proof of this theorem.


\begin{corollary}\label{cor: HGSDS}
Let $\mI$ and $\mI'$ be finitely branching interpretations (w.r.t.~$\mLP$) such that, for every $a \in \IN$, $a^\mI \equivPsp a^{\mI'}$. 
Suppose $Q \in \Phi$ and if $U \in \Phi$ then both $\mI$ and $\mI'$ are unreachable-objects-free. 
Then, for every $x \in \Delta^\mI$ and $x' \in \Delta^{\mI'}$, $x \equivPsp x'$ iff $x \equivP x'$.
\end{corollary}

This corollary follows from Theorems~\ref{theorem: H-M-3} and~\ref{theorem: H-M-0}.

\begin{example}
We show that the assumption $Q \in \Phi$ of Theorem~\ref{theorem: H-M-3} is necessary. Let $\Phi = \emptyset$, $\IN = \{a\}$, $\CN = \{A,B\}$, $\RN = \{r\}$ and let $\mI$, $\mI'$ be the interpretations specified as follows.
\begin{itemize}
\item $\Delta^\mI = \{u,v_0,v_1,v_2\}$, $a^\mI = u$, $r^\mI = \{\tuple{u,v_0}, \tuple{u,v_1}, \tuple{u,v_2}\}$, $A^\mI = \{v_1,v_2\}$, $B^\mI = \{v_2\}$,
\item $\Delta^{\mI'} = \{u,v_0,v_2\}$, $a^{\mI'} = u$, $r^{\mI'} = \{\tuple{u,v_0},\tuple{u,v_2}\}$ and $A^{\mI'} = B^{\mI'} = \{v_2\}$.
\end{itemize}
Notice that $\mI'$ is obtained from $\mI$ by deleting $v_1$. 
Observe that there are $\mLP$-comparisons between $\mI$ and $\mI'$ as well as between $\mI'$ and $\mI$, but there is no $\mLP$-bisimulations between $\mI$ and $\mI'$. In particular, $a^\mI \equivPsp a^{\mI'}$, but $a^\mI \not\equivP a^{\mI'}$.
\myend
\end{example}

The point of the above example is that, when $Q \notin \Phi$, if $v_0$, $v_1$, $v_2$ are pairwise different $r$-successors of $u$, $v_0 \leqPsp v_1$ and $v_1 \leqPsp v_2$ then the edge $\tuple{u,v_1} \in r^\mI$ is not essential for the semantics of semi-positive concepts. Also note that, when $Q \notin \Phi$, if $v$ and $v'$ are different $r$-successors of $u$ such that $v \equivPsp v'$ then the edge $\tuple{u,v'} \in r^\mI$ is not essential for the semantics of semi-positive concepts. 

Suppose $Q \notin \Phi$ and let $\mI$ be a finitely branching interpretation. We say that $\mI$ is {\em $\mLPsp$-tidy} if it is unreachable-objects-free and, for every $x,y,y',y'' \in \Delta^\mI$ and every basic role $R$ of~$\mLP$, 
\begin{itemize}
\item if $\{\tuple{x,y},\tuple{x,y'}\} \subseteq R^\mI$ and $y \equivPsp y'$ then $y = y'$,
\item if $\{\tuple{x,y},\tuple{x,y'},\tuple{x,y''}\} \subseteq R^\mI$, $y \leqPsp y'$ and $y' \leqPsp y''$ then $y = y'$ or $y' = y''$ or ($\Self \in \Phi$ and $y' = x$).
\end{itemize}

\begin{theorem} \label{theorem: H-M-4}
Suppose $Q \notin \Phi$. Let $\mI$ and $\mI'$ be finitely branching and $\mLPsp$-tidy interpretations such that, for every $a \in \IN$, $a^\mI \equivPsp a^{\mI'}$. 
Then, for every $x \in \Delta^\mI$ and $x' \in \Delta^{\mI'}$, $x \equivPsp x'$ iff there exists an $\mLP$-bisimulation $Z$ between $\mI$ and $\mI'$ such that $Z(x,x')$ holds. In particular, the relation $\{\tuple{x,x'} \in \Delta^ \mI\times \Delta^{\mI'} \mid x \equivPsp x'\}$ is an $\mLP$-bisimulation between $\mI$ and~$\mI'$.
\end{theorem}

See the appendix for a proof of this theorem.


\section{Auto-Bisimulation and Minimization}
\label{section: aut-bis min}

In this section, we recall some results of our manuscript~\cite{BSDL-arxiv}, not published in~\cite{BSDL-CSP}. 

\medskip

\noindent An~$\mLP$-bisimulation between $\mI$ and itself is called an {\em $\mLP$-auto-bisimulation of $\mI$}. An $\mLP$-auto-bisimulation of $\mI$ is said to be the {\em largest} if it is larger than or equal to ($\supseteq$) any other $\mLP$-auto-bisimulation of~$\mI$. 

\begin{proposition}
For every interpretation $\mI$, the largest $\mLP$-auto-bisimulation of $\mI$ exists and is an equivalence relation.
\myend
\end{proposition}

Given an interpretation $\mI$, by $\simLP$ we denote the largest $\mLP$-auto-bisimulation of $\mI$, and by $\equivLP$ we denote the binary relation on $\Delta^\mI$ with the property that \mbox{$x\ \equivLP\ x'$} iff $x$ is $\mLP$-equivalent to $x'$. 

\begin{theorem}\label{theorem: LABS}
For every finitely branching interpretation $\mI$, $\equivLP$ is the largest $\mLP$-auto-bisimulation of $\mI$ (i.e.\ the relations $\equivLP$ and $\simLP$ coincide).
\end{theorem}

An interpretation $\mI$ is said to be {\em minimal} among a class of interpretations if $\mI$ belongs to that class and, for every other interpretation $\mI'$ of that class, $\#\Delta^\mI \leq \#\Delta^{\mI'}$ (the cardinality of $\Delta^\mI$ is less than or equal to the cardinality of~$\Delta^{\mI'}$).

A {\em concept assertion} of $\mLP$ (resp.\ $\mLPsp$) is an expression of the form $C(a)$, where $C$ is a concept of $\mLP$ (resp.\ $\mLPsp$).
We say that an interpretation $\mI$ {\em satisfies} a concept assertion $C(a)$ if $a \in C^\mI$.

\subsection{The Case without $Q$ and $\Self$}

\newcommand{\mIq}{{\mI/_\simLP}}
\newcommand{\mIqQS}{{\mI/_\simLP^{QS}}}

The {\em quotient interpretation} $\mIq$ of $\mI$ w.r.t.\ $\simLP$ is defined as usual: 
\begin{itemize}
\item $\Delta^\mIq = \{[x]_\simLP \mid x \in \Delta^\mI\}$, where $[x]_\simLP$ is the abstract class of $x$ w.r.t.~$\simLP$
\item $a^\mIq = [a^\mI]_\simLP$, for $a \in \IN$
\item $A^\mIq = \{[x]_\simLP \mid x \in A^\mI\}$, for $A \in \CN$
\item $r^\mIq = \{\tuple{[x]_\simLP,[y]_\simLP} \mid \tuple{x,y} \in r^\mI\}$, for $r \in \RN$.
\end{itemize}


\begin{theorem} \label{theorem: JHDFA}
Suppose $\Phi \subseteq \{I,O,U\}$ and let $\mI$ be an unreachable-objects-free interpretation. If $\mIq$ is finitely branching then it is a minimal interpretation that satisfies the same concept assertions of $\mLP$ as $\mI$.
\end{theorem}

\subsection{The Case with $Q$ and/or $\Self$}

For the case when $Q \in \Phi$ or $\Self \in \Phi$, in order to obtain a result similar to Theorem~\ref{theorem: JHDFA}, we introduce QS-interpretations as follows. 

A {\em QS-interpretation} is a tuple $\mI = \langle \Delta^\mI, \cdot^\mI, \QU^\mI, \SE^\mI \rangle$, where 
\begin{itemize}
\item $\langle \Delta^\mI, \cdot^\mI\rangle$ is an interpretation, 
\item $\QU^\mI$ is a function that maps every basic role to a function $\Delta^\mI \times \Delta^\mI \to \mathbb{N}$ such that $\QU^\mI(R)(x,y) > 0$ iff $\tuple{x,y} \in R^\mI$, where $\mathbb{N}$ is the set of natural numbers, 
\item $\SE^\mI$ is a function that maps every role name to a subset of $\Delta^\mI$.
\end{itemize}

If $\mI$ is a QS-interpretation then we redefine
\begin{eqnarray*}
(\E r.\Self)^\mI & = & \{x \in \Delta^\mI \mid x \in \SE^\mI(r)\}\\
(\geq n\,R.C)^\mI & = & \{x \in \Delta^\mI \mid \Sigma\{\QU^\mI(R)(x,y) \mid C^\mI(y)\} \geq n \} \\
(\leq n\,R.C)^\mI & = & \{x \in \Delta^\mI \mid \Sigma\{\QU^\mI(R)(x,y) \mid C^\mI(y)\} \leq n \}.
\end{eqnarray*}

Other notions for interpretations remain unchanged for QS-interpretations.  

For $\mI$ being an interpretation, the {\em quotient QS-interpretation} of $\mI$ w.r.t.\ $\simLP$, denoted by $\mIqQS$, is the QS-interpretation $\mI' = \langle \Delta^{\mI'}, \cdot^{\mI'}, \QU^{\mI'}, \SE^{\mI'}\rangle$ such that: 
\begin{itemize}
\item $\langle \Delta^{\mI'}, \cdot^{\mI'}\rangle$ is the quotient interpretation of $\mI$ w.r.t.\ $\simLP$
\item for every basic role $R$ and every $x,y \in \Delta^\mI$, \[ \QU^{\mI'}(R)([x]_\simLP,[y]_\simLP) = \max_{x' \in [x]_\simLP} \#\{y' \in [y]_\simLP \mid \tuple{x',y'} \in R^\mI\} \]
\item for every role name $r$, 
\[ \SE^{\mI'}(r) = \{ [x]_\simLP \mid \tuple{x,x} \in r^\mI \}. \]
\end{itemize}

Note that, in the case when $Q \in \Phi$, we have 
\[ \QU^{\mI'}(R)([x]_\simLP,[y]_\simLP) = \#\{y' \in [y]_\simLP \mid \tuple{x,y'} \in R^\mI\}. \]

Here is a counterpart of Theorem~\ref{theorem: JHDFA}, with no restrictions on~$\Phi$:

\begin{theorem} \label{theorem: OIFJS}
Let $\mI$ be an unreachable-objects-free interpretation. If $\mIqQS$ is finitely branching then it is a minimal QS-interpretation that satisfies the same concept assertions of $\mLP$ as $\mI$.
\end{theorem}


\section{Minimization Preserving Semi-Positive Concepts}

Suppose $\Phi \subseteq \{O,U,\Self\}$ and let $\mI$ be a finitely branching interpretation such that it is also unreachable-objects-free when $U \in \Phi$. By $\TidyPsp(\mI)$ we denote the maximal $\mLPsp$-tidy sub-interpretation of $\mI$ obtained by modifying $\mI$ as follows:
\begin{itemize}
\item For each $r \in \RN$, if $\{\tuple{x,y},\tuple{x,y'}\} \subseteq r^\mI$, $y \equivPsp y'$, $y \neq y'$ and $y' \neq x$ then delete the pair $\tuple{x,y'}$ from $r^\mI$. 

\item For each $r \in \RN$, if $\{\tuple{x,y},\tuple{x,y'},\tuple{x,y''}\} \subseteq r^\mI$, $y \leqPsp y'$, $y' \leqPsp y''$, $y \not\equivPsp y'$, $y' \not\equivPsp y''$ and ($\Self \notin \Phi$ or $y' \neq x$) then delete the pair $\tuple{x,y'}$ from $r^\mI$.

\item Delete from the domain of $\mI$ all elements not reachable from any $a^\mI$ (with $a \in \IN$) via a path consisting of edges being instances of basic roles of~$\mLP$.
\end{itemize}

\begin{lemma}\label{lemma: HJWIA}
Suppose $\Phi \subseteq \{O,U,\Self\}$ and let $\mI$ be a finitely branching interpretation such that it is also unreachable-objects-free when $U \in \Phi$. Then $\TidyPsp(\mI)$ satisfies the same concept assertions of $\mLPsp$ as $\mI$.
\end{lemma}

\begin{proof}
Let $\mI' = \TidyPsp(\mI)$ and let $Z$, $Z'$ be the smallest binary relations such that the following conditions hold for every $a \in \Sigma_I$, $r \in \RN$, $x,y \in \Delta^\mI$, $x',y' \in \Delta^{\mI'}\,$:
\begin{itemize}
\item $Z(a^\mI,a^\mI)$ and $Z'(a^\mI,a^\mI)$, 
\item $Z(x,x') \land r^\mI(x,y) \land r^{\mI'}(x',y') \land y \leqPsp y' \Rightarrow Z(y,y')$, 
\item $Z'(x',x) \land r^\mI(x,y) \land r^{\mI'}(x',y') \land y' \leqPsp y \Rightarrow Z'(y',y)$. 
\end{itemize}

It is easy to see that $Z$ is an $\mLP$-comparison between $\mI$ and $\mI'$, and $Z'$ is an $\mLP$-comparison between $\mI'$ and $\mI$. Therefore, by Theorem~\ref{theorem: bs-inv-1}, $\mI'$ and $\mI$ satisfy the same concept assertions of $\mLPsp$.
\myend
\end{proof}

\begin{theorem} \label{theorem: OFNSA}
Suppose $\Phi \subseteq \{O,U,\Self\}$. Let $\mI_0$ and $\mI'_0$ be finitely branching interpretations such that they are also unreachable-objects-free when $U \in \Phi$ and they satisfy the same concept assertions of $\mLPsp$. Let $\mI = \TidyPsp(\mI_0)$, $\mI_2 = \mIq$ if $\Self \notin \Phi$, and $\mI_2 = \mIqQS$ if $\Self \in \Phi$. Then $\mI_2$ satisfies the same concept assertions of $\mLPsp$ as $\mI'_0$ and $\# \Delta^{\mI_2} \leq \# \Delta^{\mI'_0}$. 
\end{theorem}

\begin{proof}
Let $\mI' = \TidyPsp(\mI'_0)$.
By Lemma~\ref{lemma: HJWIA}, $\mI$ and $\mI'$ satisfy the same concept assertions of $\mLPsp$. 
Consequently, by Theorem~\ref{theorem: H-M-4}, there exists an $\mLP$-bisimulation between $\mI$ and $\mI'$. By Theorem~\ref{theorem: H-M-0}, it follows that $\mI$ and $\mI'$ satisfy the same concept assertions of $\mLP$. 
If $\Self \notin \Phi$ then let $\mI'_2 = \mI'/_{\sim_{\,\Phi,\mI'}}$, else let $\mI'_2 = \mI'/_{\sim_{\,\Phi,\mI'}}^{QS}$. 
By Theorems~\ref{theorem: JHDFA} and~\ref{theorem: OIFJS}, $\# \Delta^{\mI_2} = \# \Delta^{\mI'_2}$. Since $\# \Delta^{\mI'_2} \leq \# \Delta^{\mI'_0}$, it follows that $\# \Delta^{\mI_2} \leq \# \Delta^{\mI'_0}$. 
\myend
\end{proof}

\begin{theorem} \label{theorem: IUFBN}
Suppose $Q \in \Phi$. Let $\mI$ and $\mI'$ be finitely branching interpretations such that they are also unreachable-objects-free when $U \in \Phi$ and they satisfy the same concept assertions of $\mLPsp$. Then $\mI_2 = \mIqQS$ is a QS-interpretation that satisfies the same concept assertions of $\mLPsp$ as $\mI'$ and $\# \Delta^{\mI_2} \leq \# \Delta^{\mI'}$. 
\end{theorem}

\begin{proof}
Let $\mI'_2 = \mI'/_{\sim_{\,\Phi,\mI'}}^{QS}$.
By Theorem~\ref{theorem: H-M-3}, there exists an $\mLP$-bisimulation between $\mI$ and $\mI'$. By Theorem~\ref{theorem: H-M-0}, it follows that $\mI$ and $\mI'$ satisfy the same concept assertions of $\mLP$. Hence, by Theorem~\ref{theorem: OIFJS}, $\# \Delta^{\mI_2} = \# \Delta^{\mI'_2}$. Since $\# \Delta^{\mI'_2} \leq \# \Delta^{\mI'}$, it follows that $\# \Delta^{\mI_2} \leq \# \Delta^{\mI'}$.
\myend
\end{proof}

Notice that minimization of interpretations that preserves semi-positive concepts for the case when $Q \notin \Phi$ and $I \in \Phi$ is not investigated in this section.


\section{Conclusions}

We have studied bisimulation-based comparisons between interpretations in a reasonably systematic way for a large class of useful description logics and obtained novel results on ``characterizing bisimulation by semi-positive concepts'' and ``minimization preserving semi-positive concepts''. 

\medskip
\noindent
{\bf Acknowledgments.} This work was supported by the Polish National Science Centre (NCN) under Grant No.~2011/01/B/ST6/02759.



\newpage

\appendix

\section{Proofs}

\subsection*{Proof of Lemma~\ref{lemma: bs-inv-1}}

We prove this lemma by induction on the structures of $C$, $R$ and $S$.

Consider the assertion~\eqref{bs:eqC-2}.
Suppose $Z(x,x')$ and $R^\mI(x,y)$ hold. By induction on the structure of~$R$ we prove that there exists $y' \in \Delta^{\mI'}$ such that $Z(y,y')$ and $R^{\mI'}(x',y')$ hold. The base case occurs when $R$ is a role name and the assertion for it follows from~\eqref{bs:eqC}. The induction steps are given below.
\begin{itemize}
\item Case $R = \varepsilon$ is trivial.
\item Case $R = R_1 \circ R_2$, where $R_1$ and $R_2$ are roles of $\mLPspE$: We have that $(R_1\circ R_2)^\mI(x,y)$ holds. Hence, there exists $z\in \Delta^\mI$ such that $R_1^\mI(x,z)$ and $R_2^\mI(z,y)$ hold. By the inductive assumption of~\eqref{bs:eqC-2}, there exists $z' \in \Delta^{\mI'}$ such that $Z(z,z')$ and $R_1^{\mI'}(x',z')$ hold, and there exists $y' \in \Delta^{\mI'}$ such that $Z(y,y')$ and
$R_2^{\mI'}(z',y')$ hold. Since $R_1^{\mI'}(x',z')$ and $R_2^{\mI'}(z',y')$ hold, we have that $(R_1 \circ R_2)^{\mI'}(x',y') $ holds, i.e. $R^{\mI'}(x',y')$ holds.

\item Case $R = R_1 \sqcup R_2$, where $R_1$ and $R_2$ are roles of $\mLPspE$, is trivial.

\item Case $R = R_1^*$, where $R_1$ is a role of $\mLPspE$: Since $R^\mI(x,y)$ holds, there exists $x_0, \ldots, x_k \in \Delta^\mI$ such that $x_0 = x$, $x_k = y$ and, for $1 \leq i \leq k$, $R_1^\mI(x_{i-1},x_i)$ holds. Let $x'_0 = x'$. For each $1 \leq i \leq k$, since $Z(x_{i-1},x'_{i-1})$ and $R_1^\mI(x_{i-1},x_i)$ hold, by the inductive assumption of~\eqref{bs:eqC-2}, there exists $x'_i \in \Delta^\mI$ such that $Z(x_i,x'_i)$ and $R_1^{\mI'}(x'_{i-1},x'_i)$ hold. Hence, $Z(x_k,x'_k)$ and $(R_1^*)^{\mI'}(x'_0,x'_k)$ hold. Let $y'= x'_k$. Thus, $Z(y,y')$ and $R^{\mI'}(x',y')$ hold.

\item Case $R = (D?)$, where $D$ is a concept of $\mLPsp$: By the definition of $(D?)^\mI$, we have that $D^\mI(x)$ holds and $x = y$. By the inductive assumption of~\eqref{bs:eqB-2}, $D^{\mI'}(x')$ holds, and therefore $R^{\mI'}(x',x')$ holds. By choosing $y' = x'$, we have that $Z(y,y')$ and $R^{\mI'}(x',y')$ hold.

\item Case $I \in \Phi$ and $R = r^-$: The assertion for this case follows from~\eqref{bs:eqI1}.
\end{itemize}

The assertion~\eqref{bs:eqD-2} can be proved analogously as for~\eqref{bs:eqC-2} except for the case $S = (\neg C)?$, where $C$ is a concept of $\mLPsp$. The proof for this case is as follows. 
Suppose $Z(x,x')$ and $S^{\mI'}(x',y')$ hold. 
Thus, $(\neg C)^{\mI'}(x')$ holds and $x'=y'$. By the contrapositive of the inductive assumption of~\eqref{bs:eqB-2}, it follows that $(\neg C)^\mI(x)$ holds. By choosing $y=x$, $Z(y,y')$ and $S^\mI(x,y)$ hold.

Consider the assertion~\eqref{bs:eqB-2}.
Suppose $Z(x,x')$ and $C^\mI(x)$ hold, where $C$ is a concept of $\mLPsp$. We show that $C^{\mI'}(x')$ holds.
The cases when $C$ is of the form $\top$, $\bot$, $A$, $D \mor D'$ or $D \mand D'$ are trivial.
\begin{itemize}
\item Case $C = \E R.D$, where $R$ is a role of $\mLPspE$ and $D$ is a concept of $\mLPsp$: Since $(\E R.D)^\mI(x)$ holds, there exists $y \in \Delta^\mI$ such that $R^\mI(x,y)$ and $D^\mI(y)$ hold. By the inductive assumption of~\eqref{bs:eqC-2} (proved earlier), there exists $y' \in \Delta^{\mI'}$ such that $Z(y,y')$ and $R^{\mI'}(x',y')$ hold. By the inductive assumption of~\eqref{bs:eqB-2}, $D^{\mI'}(y')$ holds. Therefore, $C^{\mI'}(x')$ holds.

\item Case $C = \V S.D$, where $S$ is a role of $\mLPspV$ and $D$ is a concept of $\mLPsp$: Let $y'$ be an arbitrary element of $\Delta^{\mI'}$ such that $S^{\mI'}(x',y')$ holds. We show that $D^{\mI'}(y')$ holds. By the inductive assumption of~\eqref{bs:eqD-2} (proved earlier), there exists $y \in \Delta^\mI$ such that $Z(y,y')$ and $S^\mI(x,y)$ hold. Since $(\V S.D)^\mI(y)$ holds, it follows that $D^\mI(y)$ holds. Therefore, by the inductive assumption of~\eqref{bs:eqB-2}, it follows that $D^{\mI'}(y')$ holds.

\item Case $O \in \Phi$ and $C = \{a\}$: Since $\{a\}^\mI(x)$ holds, we have that $x = a^\mI$. By the condition~\eqref{bs:eqO0}, it follows that $x' = a^{\mI'}$. Hence $C^{\mI'}(x')$ holds.

\item Case $\Self \in \Phi$ and $C = \E r.\Self$: Since $(\E r.\Self)^\mI(x)$ holds, we have that $r^\mI(x,x)$ holds. By the condition~\eqref{bs:eqSelf}, it follows that $r^{\mI'}(x',x')$ holds. Hence $C^{\mI'}(x')$ holds.

\item Case $Q \in \Phi$ and $C = (\geq\!n\,r.D)$, where $D$ is a concept of $\mLPsp$: By the condition~\eqref{bs:eqQ}, there exists a bijection $h: \{y \mid r^\mI(x,y)\} \to \{y' \mid r^{\mI'}(x',y')\}$ such that $h \subseteq Z$. Since $(\geq\!n\,r.D)^\mI(x)$ holds, there exist pairwise different $y_1$, \ldots, $y_n \in \Delta^\mI$ such that $r^\mI(x,y_i)$ and $D^\mI(y_i)$ hold for every $1 \leq i \leq n$. For each $1 \leq i \leq n$, let $y'_i = h(y_i)$. Thus, $Z(y_i,y'_i)$ holds. By the inductive assumption of~\eqref{bs:eqB-2}, it follows that $D^{\mI'}(y'_i)$ holds. Since $r^{\mI'}(x',y')$ and $D^{\mI'}(y'_i)$ hold for $1 \leq i \leq n$, and $y_i \neq y_j$ for $1 \leq i \neq j \leq n$, it follows that \mbox{$(\geq\!n\,r.D)^{\mI'}(x')$} holds, which means $C^{\mI'}(x')$ holds.

\item Case $\{Q,I\} \subseteq \Phi$ and $C = (\geq n\,r^{-1}.D)$, where $D$ is a concept of $\mLPsp$, can be proved analogously to the above case.

\item Case $Q \in \Phi$ and $C =(\leq\!n\,r.(\neg D))$, where $D$ is a concept of $\mLPsp$: For the sake of contradiction, suppose  $C^{\mI'}(x')$ does not hold. Thus, $(\neg C)^{\mI'}(x')$ holds, which means $(\geq\!(n+1)\,r.(\lnot D))^{\mI'}(x')$ holds. By the condition~\eqref{bs:eqQ}, there exists a bijection $h: \{y \mid r^\mI(x,y)\} \to \{y' \mid r^{\mI'}(x',y')\}$ such that $h \subseteq Z$. Since $(\geq\!(n+1)\,r.(\lnot D))^{\mI'}(x')$ holds, there exist pairwise different $y'_1$, \ldots, $y'_{n+1} \in \Delta^{\mI'}$ such that $r^{\mI'}(x',y'_i)$ and $(\neg D)^{\mI'}(y'_i)$ hold for all $1 \leq i \leq n+1$. For each $1 \leq i \leq n+1$, let $y_i = h^{-1}(y'_i)$. Since $h$ is a bijection, $y_1,\ldots,y_{n+1}$ are pairwise different, and by the definition of $h$, $r^\mI(x,y_i)$ holds for every $1 \leq i \leq n+1$. 
For $1 \leq i \leq n+1$, since $(\neg D)^{\mI'}(y'_i)$ holds, by the contrapositive of the inductive assumption of~\eqref{bs:eqB-2}, it follows that $(\neg D)^\mI(y_i)$ holds. Thus, $(\neg C)^\mI(x)$ holds, which contradicts the assumption that $C^\mI(x)$ holds. Therefore, $C^{\mI'}(x')$ holds.

\item Case $\{Q,I\} \subseteq \Phi$ and $C = (\leq\!n\,r^{-1}.(\neg D))$, where $D$ is a concept of $\mLPsp$, can be proved analogously to the above case.

\item Case $U \in \Phi$ and $C = \V U.D$, where $D$ is a concept of $\mLPsp$: Let $y' \in \Delta^{\mI'}$. By the condition~\eqref{bs:eqU2}, there exists $y\in \Delta^\mI$ such that $Z(y,y')$ holds. Since $C^\mI(x)$ holds, it follows that $D^ \mI(y)$ holds. By the inductive assumption of~\eqref{bs:eqB-2}, it follows that $D^{\mI'}(y')$ holds. Hence $C^{\mI'}(x')$ holds.

\item Case $U \in \Phi$ and $C = \E U.D$, where $D$ is a concept of $\mLPsp$:
 Since $C^\mI(x)$ holds, there exists $y \in \Delta^\mI$ such that $D^\mI(y)$ holds. By the condition~\eqref{bs:eqU1}, there exists $y' \in \Delta^{\mI'}$ such that $Z(y,y')$ holds. By the inductive assumption of~\eqref{bs:eqB-2}, it follows that $D^{\mI'}(y')$ holds. Hence $C^{\mI'}(x')$  holds.
\end{itemize}


\subsection*{Proof of Theorem~\ref{theorem: H-M}}

First, suppose $Z$ is an $\mLP$-comparison between $\mI$ and $\mI'$ such that $Z(x,x')$ holds. We show that $x\leqPsp x'$. Let $C$ be an arbitrary concept of $\mLPsp$ such that $C^\mI(x)$ holds. Thus, by the assertion~\eqref{bs:eqB-2} of Lemma~\ref{lemma: bs-inv-1}, $C^{\mI'}(x')$ holds. Therefore, $x\leqPsp x'$.

Conversely, we show that $Z = \{\tuple{x,x'} \in \Delta^\mI \times \Delta^{\mI'} \mid x \leqPsp x'\}$ is an $\mLP$-comparison between $\mI$ and $\mI'$.

\begin{itemize}
\item The condition~\eqref{bs:eqA} immediately follows from the assumption of the theorem.

\item Consider the condition~\eqref{bs:eqB}. If $Z(x,x')$ and $A^\mI(x)$ hold, then by the definition of $Z$, $A^{\mI'}(x')$ holds.

\item Consider the condition~\eqref{bs:eqC}. Suppose $Z(x,x')$ and $r^\mI(x,y)$ hold. Let ${\bf S}=\{y'\in \, \Delta^{\mI'} \mid \, r^{\mI'}(x',y')\}$. We show that there exists $y' \in {\bf S}$ such that $Z(y,y')$ holds. Since $(\E r.\top)^\mI(x)$ holds and $x \leqPsp x'$, it follows that $(\E r.\top)^{\mI'}(x')$ holds. Consequently, ${\bf S} \neq \emptyset$. Since $\mI'$ is finitely branching, ${\bf S}$ must be finite. Let the elements of ${\bf S}$ be $y'_1$, \ldots, $y'_n$. For the sake of contradiction, suppose that for every $1 \leq i \leq n$, $Z(y,y'_i)$ does not hold, which means that $y \not \leqPsp y'_i$. Thus, for every  $1 \leq i \leq n$, there exists a concept $C_i$ of $\mLPsp$ such that $C_i^\mI(y)$ holds, but $C_i^{\mI'}(y')$ does not. Let $C = \E r.(C_1 \mand\ldots\mand C_n)$. Thus, $C^\mI(x)$ holds, but $C^{\mI'}(x')$ does not. This contradicts $x \leqPsp x'$. Hence, there exists $y'_i \in {\bf S}$ such that $Z(y,y'_i)$ holds.

\item Consider the condition~\eqref{bs:eqD}. Suppose $Z(x,x')$ and $r^{\mI'}(x',y')$ hold. Let ${\bf S} = \{y \in \Delta^\mI \mid r^\mI(x,y)\}$. We show that there exists $y\in {\bf S}$ such that $Z(y,y')$ holds. For the sake of contradiction, suppose ${\bf S}=\emptyset$. Thus, $(\V r.\bot)^\mI(x)$ holds. Since $x \leqPsp x'$, it follows that $C^{\mI'}(x')$ holds, and hence $\bot^{\mI'}(y')$ holds, which is a contradiction. Therefore, ${\bf S} \neq \emptyset$. 
Since $\mI$ is finitely branching, ${\bf S}$ must be finite. Let $y_1, \ldots , y_n $  be all the elements of ${\bf S}$. For the sake of contradiction, suppose that for every $1 \leq i \leq n$, $Z_i(y_i,y')$ does not hold, i.e. $y_i \not \leqPsp y'$. Thus, for every $1\leq i \leq n$, there exists a concept $C_i$ of $\mLPsp$ such that $C_i^\mI(y_i)$ holds, but $C_i^{\mI'}(y')$ does not. Let $C=\V r.(C_1 \sqcup \ldots \sqcup C_n)$. Clearly, $C^\mI(x)$ holds, but $C^{\mI'}(x')$ does not. This contradicts $x \leqPsp x'$. Hence, there exists $y_i \in {\bf S}$ such that $Z(y_i,y')$ holds.

\item The conditions~\eqref{bs:eqI1} and~\eqref{bs:eqI2} can be proved analogously as for the conditions~\eqref{bs:eqC} and~\eqref{bs:eqD}, respectively.

\item Consider the condition~\eqref{bs:eqO0} and the case $O \in \Phi$. Suppose $Z(x,x')$ holds and $x = a^\mI$. Since $\{a\}^\mI(x)$ holds and $x \leqPsp x'$, it follows that $\{a\}^{\mI'}(x')$ holds. Therefore, $x' = a^{\mI'}$.

\item Consider the condition~\eqref{bs:eqQ} and the case $Q \in \Phi$. Suppose $Z(x,x')$ holds, i.e., $x \leqPsp x'$. Let ${\bf S} = \{y \in \Delta^\mI \mid r^\mI(x,y)\}$ and ${\bf S}'=\{y'\in \Delta^{\mI'}\mid \,r^{\mI'}(x',y')\}$. Since $\mI$ and $\mI'$ are finitely branching, ${\bf S}$ and ${\bf S'}$ must be finite. Let $m = \# \mathbf{S}$ and $n = \#\mathbf{S'}$. We first show that $m = n$. If $m > n$ then $x \in (\geq\!m\,r.\top)^\mI$ and $x' \notin (\geq\!m\,r.\top)^{\mI'}$, which contradicts $x \leqPsp x'$. If $m < n$ then $x \in (\leq\!m\,r.\lnot\bot)^\mI$ and $x' \notin (\leq\!m\,r.\lnot\bot)^{\mI'}$, which contradicts $x \leqPsp x'$. Therefore $m = n$. 
Let $\mathbf{S} = \{y_1,\ldots,y_m\}$. We can try to construct a bijection $h : \mathbf{S} \to \mathbf{S'}$ such that $h \subseteq Z$ as follows. For each $i$ from 1 to $m\,$:
  \begin{itemize}
  \item If there exists $y' \in \mathbf{S'} \setminus \{h(y_1),\ldots,h(y_{i-1})\}$ such that $Z(y_i,y')$ holds then set $h(y_i) := y'$ and continue with the next $i$. 
  \item Consider the other case. By the assertion~\eqref{bs:eqC}, there exists $y' \in \mathbf{S'}$ such that $Z(y_i,y')$ holds. Nondeterministically choose $1 \leq j < i$ such that $h(y_j) = y'$, exchange $y_i$ and $y_j$, and go back to the previous step.
  \end{itemize}
For the sake of contradiction, suppose that for some $1 \leq i \leq m$, every possible run of the above loop does not terminate. There must exist $\mathbf{S}_0 \subseteq \{y_1,\ldots,y_{i-1}\}$ such that, for every $y \in \mathbf{S}_0 \cup \{y_i\}$ and every $y' \in \mathbf{S'}$, if $Z(y,y')$ holds then $y' \in h(\mathbf{S}_0)$. Let $\mathbf{S}_0 \cup \{y_i\} = \{u_1,\ldots,u_h\}$ and $\mathbf{S'} \setminus h(\mathbf{S}_0) = \{v_1,\ldots,v_k\}$. We have $h + k = m+1$, hence $h > m- k$. For each $1 \leq i \leq h$ and $1 \leq j \leq k$, since $Z(u_i,v_j)$ does not hold, there exists a concept $C_{i,j}$ of $\mLPsp$ such that $C_{i,j}^\mI(u_i)$ holds, but $C_{i,j}^{\mI'}(v_j)$ does not. For $1 \leq i \leq h$, let $C_i = C_{i,1} \mand \ldots \mand C_{i,k}$. Then let $C = C_1 \mor\ldots\mor C_h$. Observe that $\{u_1,\ldots,u_h\} \subseteq C^\mI$ and $\{v_1,\ldots,v_k\} \cap C^{\mI'} = \emptyset$. Thus, $x \in (\geq\!h\,r.C)^\mI$ and $x' \notin (\geq\!h\,r.C)^{\mI'}$, which contradicts the assumption that $x \leqPsp x'$. Therefore, there exists a bijection $h : \mathbf{S} \to \mathbf{S'}$ such that $h \subseteq Z$.  

\item The condition~\eqref{bs:eqQI} can be proved analogously as for the condition~\eqref{bs:eqQ}.

\item Consider the condition~\eqref{bs:eqU1} and the case $U \in \Phi$. By the assumption of this case, either $\IN \neq \emptyset$ and both $\mI$, $\mI'$ are finite, or both $\mI$, $\mI'$ are unreachable-objects-free. 
  \begin{itemize}
  \item Case $\IN \neq \emptyset$ and both $\mI$, $\mI'$ are finite: Let $x \in \Delta^\mI$ and let $x'_1, \ldots , x'_n $ be all the elements of $\Delta^{\mI'}$. For the sake of contradiction, suppose that for every $1 \leq i\leq n$, $x \not \leqPsp x'_i$. Thus, for every $1 \leq i\leq n$, there exists a concept $C_i$ of $\mLPsp$ such that $C_i^\mI(x)$ holds, but $C_i^{\mI'}(x'_i)$ does not. Let $C=C_1 \mand \ldots \mand C_n$ and $a\in \Sigma_I$. Since $C^\mI(x)$ holds, $(\E U.C)^\mI (a^\mI)$ also holds, but $(\E U.C)^{\mI'} (a^{\mI'})$ does not, which contradicts the assumption $a^\mI \leqPsp a^{\mI'}$.
  \item Case both $\mI$, $\mI'$ are unreachable-objects-free: The condition~\eqref{bs:eqU1} follows from the conditions \eqref{bs:eqA}, \eqref{bs:eqC} and \eqref{bs:eqD}. 
  \end{itemize}

\item The condition~\eqref{bs:eqU2} can be proved analogously as for the condition~\eqref{bs:eqU1}.

\item Consider the condition~\eqref{bs:eqSelf} and the case $\Self \in \Phi$. Suppose $Z(x,x')$ and $r^\mI(x,x)$ hold. Since $(\E r.\Self)^\mI(x)$ holds and $x \leqPsp x'$, it follows that $(\E r.\Self)^{\mI'}(x')$ holds. Hence, $r^{\mI'}(x',x')$ holds.
\end{itemize}


\subsection*{Proof of Theorem~\ref{theorem: H-M-3}}

If $Z$ is an $\mLP$-bisimulation between $\mI$ and $\mI'$ such that $Z(x,x')$ holds then, by Theorem~\ref{theorem: H-M-0}, $x \equivP x'$, and hence $x \equivPsp x'$. For the remaining assertions of the current theorem, we show that $Z = \{\tuple{x,x'} \in \Delta^\mI \times \Delta^{\mI'} \mid x \equivPsp x'\}$ is an $\mLP$-bisimulation between $\mI$ and $\mI'$.

\begin{itemize}
\item The condition~\eqref{bs:eqA} immediately follows from the assumption of the theorem.

\item Consider the condition~\eqrefp{bs:eqB}. Suppose $Z(x,x')$ holds. By the definition of $Z$, $A^\mI(x)$ holds iff $A^{\mI'}(x')$ holds.

\item Consider the condition~\eqrefp{bs:eqO0} and the case $O \in \Phi$. Suppose $Z(x,x')$ holds. Thus, $\{a\}^\mI(x)$ holds iff $\{a\}^{\mI'}(x')$ holds. That is, $x = a^\mI$ iff $x' = a^{\mI'}$.

\item Consider the condition~\eqrefp{bs:eqSelf} and the case $\Self \in \Phi$. Suppose $Z(x,x')$ holds. Thus, $(\E r.\Self)^\mI(x)$ holds iff $(\E r.\Self)^{\mI'}(x')$ holds. That is, $r^\mI(x,x)$ holds iff $r^{\mI'}(x',x')$ holds.

\item Consider the condition~\eqref{bs:eqQ} and the case $Q \in \Phi$. Suppose $Z(x,x')$ holds, i.e., $x \equivPsp x'$. Let ${\bf S} = \{y \in \Delta^\mI \mid r^\mI(x,y)\}$ and ${\bf S}'=\{y'\in \Delta^{\mI'}\mid \,r^{\mI'}(x',y')\}$. Since $\mI$ and $\mI'$ are finitely branching, ${\bf S}$ and ${\bf S'}$ must be finite. As shown in the proof of Theorem~\ref{theorem: H-M}, there exists a bijection $h : \mathbf{S} \to \mathbf{S'}$ such that, if $h(y) = y'$ then $y \leqPsp y'$. Analogously, there exists a bijection $h' : \mathbf{S'} \to \mathbf{S}$ such that, if $h'(y') = y$ then $y' \leqPsp y$. Therefore, there must exist a bijection $h_2 : \mathbf{S} \to \mathbf{S'}$ such that, if $h_2(y) = y'$ then $y \equivPsp y'$. 

\item The condition~\eqref{bs:eqQI} can be proved analogously as for the condition~\eqref{bs:eqQ}.

\item The conditions~\eqref{bs:eqC} and~\eqref{bs:eqD} follow from the condition~\eqref{bs:eqQ}.
\item The conditions~\eqref{bs:eqI1} and~\eqref{bs:eqI2} follow from the condition~\eqref{bs:eqQI}.
\item Consider the conditions~\eqref{bs:eqU1} and~\eqref{bs:eqU2} and the case $U \in \Phi$. By assumption, both $\mI$ and $\mI'$ are unreachable-objects-free. The condition~\eqref{bs:eqU1} follows from the conditions \eqref{bs:eqA}, \eqref{bs:eqC} and \eqref{bs:eqD}. Analogously, the condition~\eqref{bs:eqU2} also holds.
\end{itemize}


\subsection*{Proof of Theorem~\ref{theorem: H-M-4}}

Let $Z = \{\tuple{x,x'} \in \Delta^\mI \times \Delta^{\mI'} \mid x \equivPsp x'\}$. Analyzing the proof of Theorem~\ref{theorem: H-M-3}, it suffices to show that the condition~\eqref{bs:eqC} holds (the conditions~\eqref{bs:eqD}, \eqref{bs:eqI1} and~\eqref{bs:eqI2} can be proved in a similar way). 
Suppose $Z(x,x') \land r^\mI(x,y)$ holds. We show that there exists $y'$ such that $Z(y,y') \land r^{\mI'}(x',y')$ holds. This is trivial for the case when $\Self \in \Phi$ and $y = x$. So, suppose $\Self \notin \Phi$ or $y \neq x$. Analogously to the proof of Theorem~\ref{theorem: H-M}, it can be shown that there exists $y'_2 \in \Delta^{\mI'}$ such that $r^{\mI'}(x',y'_2)$ holds and $y \leqPsp y'_2$. Dually, there exists $y'_1 \in \Delta^{\mI'}$ such that $r^{\mI'}(x',y'_1)$ holds and $y'_1 \leqPsp y$. Similarly, there exist $y_1, y_2 \in \Delta^\mI$ such that $r^\mI(x,y_1)$ and $r^\mI(x,y_2)$ hold, $y_1 \leqPsp y'_1$ and $y'_2 \leqPsp y_2$. Hence $y_1 \leqPsp y \leqPsp y_2$. Since $\mI$ is $\mLPsp$-tidy, either $y = y_1$ or $y = y_2$. Since $y_1 \leqPsp y'_1 \leqPsp y$ and $y \leqPsp y'_2 \leqPsp y_2$, it follows that $y \equivPsp y'_1$ or $y \equivPsp y'_2$, which completes the proof.


\end{document}